\documentclass{article}

\usepackage[letterpaper,top=3cm,bottom=3cm,left=3cm,right=3cm,marginparwidth=1.75cm]{geometry}

\usepackage{xspace}
\usepackage{multirow}
\usepackage{amsmath,amsthm,amssymb}
\usepackage{cite}
\usepackage{comment}
\usepackage{graphicx}
\usepackage{listings}
\usepackage{mathtools}
\usepackage{url}
\usepackage{xcolor}
\usepackage{booktabs}

\usepackage[ruled,titlenotnumbered,noend,noline]{algorithm2e}

\newcommand{\tuple}[1]{\ensuremath{\left\langle #1 \right\rangle}\xspace}
\newcommand{\size}[1]{\ensuremath{\left| #1 \right|}\xspace}
\newcommand{\set}[1]{\ensuremath{\left\{ #1 \right\}}\xspace}
\newcommand{\cp}{\ensuremath{\textbf{P}}\xspace}
\newcommand{\cnp}{\ensuremath{\textbf{NP}}\xspace}
\newcommand{\cnpc}{\ensuremath{\textbf{NP}\text{-complete}}\xspace}
\newcommand{\cnph}{\ensuremath{\textbf{NP}\text{-hard}}\xspace}
\newtheorem{theorem}{Theorem}
\newtheorem{definition}{Definition}

\title{Minimizing the Number of Roles in Bottom-Up Role-Mining using Maximal Biclique Enumeration}

\author{Mahesh Tripunitara\\
tripunit@uwaterloo.ca\\
University of Waterloo\\
Canada}

\begin{document}
\maketitle

\begin{abstract}
  Bottom-up role-mining is the determination of a set of roles given as input a set of users and the permissions those users possess. It is well-established in the research literature, and in practice, as an important problem in information security. A natural objective that has been explored in prior work is for the set of roles to be of minimum size. We address this problem for practical inputs while reconciling foundations, specifically, that the problem is \cnph. We first observe that an approach from prior work that exploits a sufficient condition for an efficient algorithm, while a useful first step, does not scale to more recently proposed benchmark inputs. We propose a new technique: the enumeration of maximal bicliques. We point out that the number of maximal bicliques provides a natural measure of the hardness of an input. We leverage the enumeration of maximal bicliques in two different ways. Our first approach addresses more than half the benchmark inputs to yield exact results. The other approach is needed for hard instances; in it, we identify and adopt as roles those that correspond to large maximal bicliques. We have implemented all our algorithms and carried out an extensive empirical assessment, which suggests that our approaches are promising. Our code is available publicly as open-source.
\end{abstract}

\section{Introduction}\label{sec:intro}
Access control is used to determine the actions a user may perform on a resource.  It is recognized in research and practice as an essential aspect of information security. Underlying access control is a \emph{policy}, which expresses the accesses that should be allowed. Such a policy may say, for example, that a user Alice may be allowed to read a file, but not write it. Two well-established syntaxes for the policy are the access matrix \cite{hru76} and role-based access control \cite{sandhu96}.

\begin{figure}[h]
\centering
\includegraphics[width=0.67\linewidth]{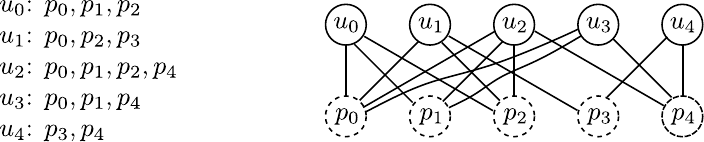}
\caption{An access matrix, as it pertains to this work, as permissions a user possesses, and as a bipartite graph. Users $u_0$ -- $u_4$ are shown as solid circles, and permissions $p_0$ -- $p_4$ as dotted circles.}
\label{fig:irreducible:am}
\end{figure}

The former, as it pertains to this work, is a \emph{bipartite graph} --- Figure \ref{fig:irreducible:am} shows an example. An (undirected) graph is a pair $\tuple{V, E}$ of vertices $V$ and edges $E$, where each member of $E$ is a subset of $V$ of size exactly two. A bipartite graph is a graph in which the vertices can be partitioned into two, call them $V_1$ and $V_2$, such that given any edge $\tuple{u_1, u_2}$, either $u_1 \in V_1$ and $u_2 \in V_2$, or $u_1 \in V_2$ and $u_2 \in V_1$. An access matrix is a bipartite graph in which the partition of the set of vertices is to a set of users $U$ and a set of permissions $P$, and an edge $\tuple{u,p}$, where $u\in U, p\in P$, expresses that the user $u$ possesses permission $p$. Other ways of perceiving an access matrix, such as one in which users (or subjects) are assigned rights over objects \cite{hru76}, can easily be encoded in this syntax --- simply adopt as a permission a pair $\tuple{\text{object}, \text{right}}$.

\begin{figure*}
\centering
\includegraphics[width=0.75\linewidth]{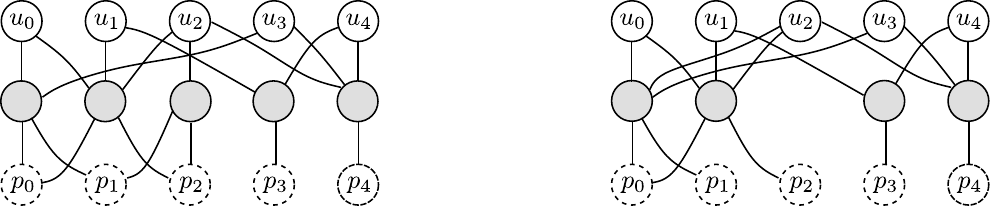}
\caption{Two role-based policies that are equivalent in authorizations to the access matrix of Figure \ref{fig:irreducible:am}. Roles are shown as shaded circles. The policy to the right has four roles, which happens to be the minimum possible for this access matrix.}
\label{fig:irreducible:rbac}
\end{figure*}

In Role-Based Access Control (RBAC), a user is assigned \emph{roles} from a set $R$, each of which is assigned permissions. It is a special case of a \emph{tripartite} graph in which the edges are between $U$ and $R$, and $R$ and $P$ only; there are no edges between $U$ and $P$. A user possesses those permissions that are assigned to the roles to which the user is assigned. The two policies shown in Figure \ref{fig:irreducible:rbac} are equivalent to one another, and in turn, equivalent to the access matrix in Figure \ref{fig:irreducible:am} in that in all three policies, every user possesses the same permissions. The policy to the right in Figure \ref{fig:irreducible:rbac}, however, is one in which the set of roles is of smallest size; it turns out that for the access matrix of Figure \ref{fig:irreducible:am}, we need at least four roles in any role-based policy that is equivalent from the standpoint of authorizations.

\emph{Role mining}, which is the focus of our work, is the problem, given as input an access matrix such as the one in Figure \ref{fig:irreducible:am}, of determining a role-based policy such as one of the ones in Figure \ref{fig:irreducible:rbac}. Apart from computing a set of roles and relationships that preserves the authorizations of the input access matrix, a basic \emph{soundness} requirement, one typically associates additional goodness criteria with the problem. The goodness criterion that is the focus of our work is minimization of the set of roles that is output. Thus, in our work, the role-based policy to the right in Figure \ref{fig:irreducible:rbac} would be correct, but the policy to the left would not.

Role-mining is motivated by the benefits of using RBAC for access control \cite{sandhu96}. To use RBAC, one must first express one's access control policy as an RBAC policy. Role-mining is an approach to achieve this. To our knowledge, the first work that explicitly addresses such a problem as role mining is that of Vaidya et al.~\cite{vaidya06}. That work, in turn, credits Coyne \cite{coyne96} as pointing out that the identification of roles is the first step in realizing RBAC, and the work of Gallagher et al.~\cite{gallagher02} for pointing out that the identification of roles is the most expensive step in realizing RBAC. That work also distinguishes a top-down approach to the identification of roles, from a bottom-up approach; the latter is the focus of our work.

While the work of Vaidya et al.~\cite{vaidya06} does not explicitly mention the minimization of the number of roles as a goodness criterion, work that followed soon after proposes exactly such a criterion \cite{vaidya07}. The work of Ene et al.~\cite{ene08} is the first, to our knowledge, to identify the correspondence between minimizing the number of roles, and minimizing the size of a \emph{biclique cover} of a bipartite graph. Given a bipartite graph $\tuple{V_1 \cup V_2, E}$, a \emph{biclique} is subsets $W_1 \subseteq V_1$, $W_2 \subseteq V_2$ such that given any $w_1\in W_1$ and $w_2\in W_2$, it is the case that $\tuple{w_1, w_2}\in E$. Equivalently, rather than as vertices, we can perceive a biclique as the corresponding set of edges. Figure \ref{fig:irreducible:biclique} shows an example of a biclique for the access matrix of Figure \ref{fig:irreducible:am}. A biclique cover, $\mathcal{C}$, given a bipartite graph, is a set of bicliques $\set{C_1, \ldots, C_k}$ such that every edge in the graph is in some $C_i$ in $\mathcal{C}$.

The biclique we show in Figure \ref{fig:irreducible:biclique} happens to be \emph{maximal}: the addition of any edges that exist in the original access matrix no longer yields a biclique. We point out also that the biclique in that figure corresponds exactly to a role in the policy to the right in Figure \ref{fig:irreducible:rbac} that minimizes the number of roles --- the second role from the left. Indeed, every role in that policy corresponds to a maximal biclique; this is by intent, as we discuss in Section \ref{sec:enum}.

\begin{figure}
\centering
\includegraphics[width=0.4\linewidth]{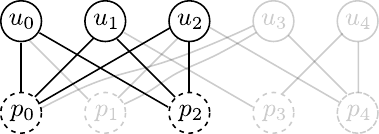}
\caption{The vertices and edges in bold show a maximal biclique for the access matrix from Figure \ref{fig:irreducible:biclique}; the other components are shown faded. Removing any vertex from $\set{u_0, u_1, u_2}\cup\set{p_0, p_2}$ still yields a biclique; however, no edge in the original graph can be added to still yield a biclique. This maximal biclique corresponds to the second role from the left in the role-based policy to the right in Figure \ref{fig:irreducible:rbac}.}
\label{fig:irreducible:biclique}
\end{figure}

From the standpoint of an efficient algorithm for the role minining problem that minimizes the numer of roles then, we ask how computationally hard the problem of computing a biclique cover of minimum size given a bipartite graph is. This problem is known to be \cnph, with a decision problem with which the optimization problem is related polynomially in \cnp (therefore \cnpc) \cite{garey1979}. Furthermore, it is in the hardest subclass of \cnp from the standpoint of efficient approximation  --- under the customary assumption $\cp \not= \cnp$, no polynomial-time algorithm can approximate to within a factor $n^\delta$ for any $\delta \in (0,1)$, where $n$ is the size of the input \cite{ene08}. Thus, we expect no efficient general-purpose algorithm to exist; not even one that approximates well. However, we know that instances of such problems that arise in practice are not necessarily hard, and can possess features that can be exploited for an efficient algorithm. Indeed, the work of Ene et al.~\cite{ene08} does exactly this for a set of benchmarks it proposes. That work ingeniously identifies and exploits a sufficient condition to trim the input so that on the benchmarks it proposes, we are left with at worst a graph of small size to address further (see Section \ref{sec:ene}).

\paragraph*{What is new?} Our motivation for revisiting the problem of role-mining for a minimum-sized set of roles is the more recent publication of a new set of benchmarks for role-mining \cite{anderer21,rmplib}.  The prior approach of Ene et al.~\cite{ene08} is highly effective for the initial set of benchmarks they propose. A natural question that arises is: is it effective for the new benchmarks as well? We point out, in this context, that Anderer et al.~\cite{anderer21} make the case that while the original benchmarks are indeed reflective of some real-world instances, their new benchmarks also capture real-world instances, but ones with different properties. If the answer to the above question is `no', i.e., that the approach of Ene et al.~\cite{ene08} is not effective on the new benchmarks, then is there an approach that is? Or do the new benchmarks comprise instances that are inherently hard? Put differently, under the customary assumption $\cp \not= \cnp$, we know that hard instances of the problem exist; are the new benchmarks such hard instances? Is there a meaningful way to identify whether a particular instance, one from benchmarks or one that arises in the real-world, is hard?

\renewcommand*{\thefootnote}{\fnsymbol{footnote}}

\paragraph*{Our work and contributions} We answer the above questions. We begin by revisiting the approach of Ene et al.~\cite{ene08} that is highly effective on their original benchmark. We claim two modest contributions with respect to it: a different exposition than in the original work, and an implementation that we make available as open-source \cite{mycode} \footnote[2]{\noindent We contacted the authors of \cite{ene08} and were told that an implementation is no longer available from them.}.  We establish that on the more recently proposed set of benchmarks \cite{anderer21,rmplib}, the approach is not effective by itself (see Table \ref{tab:dom:new} in Section \ref{sec:ene} and the associated discussion). While that algorithm can indeed be used as a first step to reduce the size of the problem to some extent, we need an approach for the rather sizeable problem that remains. 

We propose exactly such an approach. In our approach, we first observe that every role in a minimum-sized set can correspond to a maximal biclique, rather than a biclique only. This observation provides us with a natural measure of the hardness of an instance of the problem once we have run the algorithm of Ene et al.~\cite{ene08} as a first step: the total number of maximal bicliques in the access matrix. If that number is small, then we deem the instance to be ``easy''; otherwise, the instance is ``hard''. (In the worst-case, the number of maximal bicliques is exponential in the size of the input; however, it is logarithmic in the size of the set of all bicliques.)

If an instance is indeed easy, i.e., has relatively few maximal bicliques, then our approach is to enumerate all possible maximal bicliques, and have a constraint solver choose a set of minimum size from amongst them. Note that this does not necessarily render the problem that the constraint solver addresses easy in a computational sense --- we reason about this in Section \ref{sec:enum}. We have implemented this approach, and our results are that it addresses more than half the inputs in the new benchmark, i.e., yields us a minimum-sized set of roles. For the remainder of the benchmark-inputs, i.e., the ones that are hard, we do not expect there to exist an efficient approach, and are reliant on heuristics. We seek a heuristic that (i) demonstrates tangible progress and does not get ``stuck'', as a constraint-solver can \cite{solverstuck}, and, (ii) works well in practice, i.e., yields relatively few roles for the benchmark inputs. Towards this, we exploit the approach of enumerating maximal bicliques differently --- see Section \ref{sec:hard}. Our heuristic outperforms a prior heuristic, also from Ene et al.~\cite{ene08}, on all but a single benchmark-input. We claim our implementations of these heuristics as a contribution as well.

\paragraph*{Our code, the benchmarks, and their availability} Part of our contribution is an implementation which we release as open-source \cite{mycode}. Our empirical assessment is based on this implementation, and therefore is fully reproducible. We have assessed across two sets of benchmarks. One is from Ene et al.~\cite{ene08}. We do not have permission to redistribute it; see, however, the acknowledgements towards the end of this paper. There was a slight difference in what we acquired and the benchmark inputs as reported by Ene et al.~\cite{ene08}: an input instance named {customer} is missing, but two other instances, {mailer} and {univ}, have been added. 

The other, more recent, set of benchmark inputs is available publicly \cite{rmplib}. A useful attribute of the newer benchmarks as it pertains to our work is that the input access matrices appear to have been generated from existing role-based policies. As a consequence, most input instances come with a known upper-bound for the number of roles. This immediately gives us a way of assessing the quality of a technique to minimize the number of roles. We do exactly this --- see, for example, Table \ref{tab:maxsets} in Section \ref{sec:enum}, and Table \ref{tab:hard} in Section \ref{sec:hard}.

\paragraph*{Experimental setup and software} We ran all our experiments on a 64-core 4023S-TRT rack-mounted server with 256 GB RAM. It runs the Ubuntu 20.04 LTS operating system. Our code is written in Python 3. Our constraint solver is Gurobi version 11 \cite{gurobi}.

\paragraph*{Layout} The remainder of this paper is organized as follows. In the next section, we discuss related work. In Section \ref{sec:ene}, we describe the prior algorithm of Ene et al.~\cite{ene08} employing an exposition that is different from the original, and our empirical results on its effectiveness. In Section \ref{sec:enum}, we discuss our new approach based on enumerating maximal bicliques. We include in that section a discussion about the branch-and-price method from prior work on constraint-solving \cite{mehrotra96,barnhart98} of which our approach is an adaptation. In Section \ref{sec:hard}, we discuss a different way of leveraging maximal bicliques for hard instances. We conclude with Section \ref{sec:concl}, where we discuss also future work.

\section{Related Work}\label{sec:rel}
To our knowledge, the first work to call out the role mining problem is that of Vaidya et al.~\cite{vaidya06}. That work proposes a particular approach that it calls subset enumeration for generating a set of roles. Notwithstanding the occurrence of the term ``enumeration'' in both our title and theirs, their work is of a quite different nature: minimizing the number of roles is not an explicit objective; rather the intent is to generate ``natural'' roles that somehow meaningfully group permissions together. Work that appeared soon after \cite{vaidya07, ene08}, however, does explicitly propose minimization of the number of roles as an objective.

The work of Vaidya et al.~\cite{vaidya07} proposes also two other variants of the role-mining problem that are the beyond the scope of this work. Also, via a reduction from set basis, which is known to be \cnph \cite{garey1979}, it establishes that the problem of determining a minimum sized set of roles is \cnph, and a corresponding decision version is in \cnp (and is therefore \cnpc). It proposes also an algorithm based on tiling for the problem. The work of Ene et al.~\cite{ene08} identifies that minimizing the number of roles is exactly the problem of computing a minimum-sized biclique cover of a bipartite graph. Also, it identifies a sufficient condition and corresponding polynomial-time algorithm, which we revisit and employ as a first step in our approach --- see Section \ref{sec:ene}. It presents its insights via a reduction to clique partition \cite{garey1979}; we adopt the more direct rendition of bicliques. To our knowledge, the only subsequent work to target minimization of roles is that of Anderer et al.~\cite{anderer20}, which proposes an evolutionary algorithm. Their results are included as the best-known bound from an algorithm in the new benchmark suite \cite{rmplib}. Our results are stronger (i.e., smaller number of roles) on several of the new benchmark inputs for which they report a best-known bound.

\begin{figure*}
\centering
\includegraphics[width=0.75\linewidth]{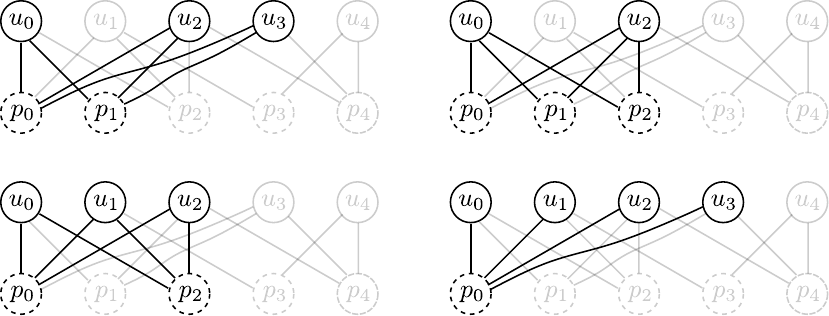}
\caption{The four maximal bicliques of which the edge $\tuple{u_0, p_0}$ is a member in the access matrix of Figure \ref{fig:irreducible:am} shown in bold. The edge $\tuple{u_2, p_0}$ is a member of all of them, and therefore, by Definition \ref{def:dominates}, dominates $\tuple{u_0, p_0}$.}
\label{fig:irreducible:oobicliques}
\end{figure*}

Since the first piece of work, there has been a proliferation of work on role-mining; Mitra et al.~\cite{mitra16} provide a survey. Most work has been in identifying meaningful goodness criteria other than role-minimization, such as roles with semantic meaning \cite{molloy09b}, roles that minimize the number of edges \cite{vaidya09} and RBAC policies that possess so called secrecy resilience \cite{guo22}. Molloy et al.~\cite{molloy09} propose ways of evaluating such varied role-mining algorithms, and corresponding ways to generate test-data. We are not aware of any benchmark datasets that have been published from these techniques, nor do we know whether the data generated is meaningful from the standpoint of instances that arise in the real-world. More recent work seeks to generalize the objective in role mining by incorporating a notion called noise, which is a controlled way of relaxing the soundness criterion which we mention in Section \ref{sec:intro} \cite{crampton22}.

There is recent work that leverages role-mining to discover bugs in smart contracts \cite{liu22}. The intent of using role-mining in such work is to discover the access control model in a role-based syntax, to then check for unauthorized accesses. Recent work that is more relevant to our work is that of Anderer et al.~\cite{anderer21} that proposes a new set of benchmarks for role-mining, which has been made available publicly \cite{rmplib}. This is exactly a set of benchmarks on which we evaluate our approaches. That work discusses the manner in which the benchmarks are generated, and the ways in which this new set of benchmark inputs is difference from those of Ene et al.~\cite{ene08}. The intent is not to explicitly target approaches that minimize the number of roles; rather, the intent behind the new set of benchmarks is to meaningfully capture real-world instances of the role-mining problem. As most instances in the benchmark appear to have been generated from role-based policies, they come with an upper-bound for the number of roles the output of a minimization algorithm can be checked against. We do exactly this in this work --- see Table \ref{tab:maxsets} in Section \ref{sec:enum}, and Tables \ref{tab:hard} and \ref{tab:hardest} in Section \ref{sec:hard}.

\section{A Prior Approach and Its Performance}
\label{sec:ene}

In this section, we revisit a sufficient condition for an efficient algorithm from prior work \cite{ene08}, and the algorithm that results from it. The sufficient condition relates to whether an edge in the input access matrix \emph{dominates} another edge. We discuss also the effectiveness of the algorithm on benchmark inputs. Apart from our assessment on the new benchmark inputs, we claim two somewhat modest contributions in this section. One is that our exposition is different from that of Ene et al.~\cite{ene08} --- that work first adopts a reduction to the clique partition problem for an undirected graph \cite{garey1979}, and then argues the soundness of the approach from the standpoint of cliques, i.e., vertices each pair of whom has an edge in the resultant graph. Our exposition stays with the notion of bicliques, which corresponds more directly to our problem: a role in an output RBAC policy from role-mining is exactly a biclique in the input access matrix. Our other contribution is an open-source implementation \cite{mycode}. Our implementation deals directly with bicliques in the input access matrix, rather than first reducing to clique partition.

\begin{definition}[Dominates]\label{def:dominates}
We say that an edge $\tuple{u_d, p_d}$ in an access matrix dominates an edge $\tuple{u, p}$ if: $\tuple{u_d, p_d}$ is a member of every maximal biclique of which $\tuple{u,p}$ is a member.
\end{definition}

Figure \ref{fig:irreducible:oobicliques} shows an example for the access matrix of Figure \ref{fig:irreducible:am}. As the edge $\tuple{u_2, p_0}$ is a member of every maximal biclique of which $\tuple{u_0, p_0}$ is a member, the former dominates the latter.

The benefit from identifying a dominator $\tuple{u_d, p_d}$ of an edge $\tuple{u, p}$ is articulated in Theorem \ref{thm:dominates} below: $u_d$ can acquire $p_d$ through any role $r$ through which $u$ acquires $p$. Before we state and prove the theorem, we recall what the subgraph induced by $S\subseteq V$ of an undirected graph $\tuple{V, E}$ is. Such a subgraph is itself an undirected graph, $\tuple{S, E_s}$, where $E_s = \set{\tuple{u,v} : \tuple{u,v} \in E\text{ and } \set{u,v} \subseteq S}$. That is, the subgraph induced by $S\subseteq V$ comprises all the edges in the original graph both of whose endpoint vertices are in $S$. 

\begin{theorem}\label{thm:dominates}
Suppose the edge $\tuple{u_d, p_d}$ dominates $\tuple{u,p}$ in the input access matrix. Then, there exists an RBAC policy with the minimum number of roles in which $u_d$ acquires $p_d$ through the same role through which $u$ acquires $p$.
\end{theorem}
\begin{proof}
Let $A$ be an RBAC policy with the minimum number of roles, and suppose $u$ acquires $p$ through the role $r$ in $A$. Suppose the users assigned to $r$ in $A$ are $\set{u_1, \ldots, u_m}$ and the permissions to which $r$ is assigned are $\set{p_1, \ldots, p_n}$. If $\set{u_d, p_d}\subseteq \set{u_1, \ldots, u_m, p_1, \ldots, p_n}$, then we have nothing left to prove. Otherwise, we prove by construction. We observe that the vertices $\set{u_1, \ldots, u_m, p_1, \ldots, p_n}$ induce a biclique, denote it $C$, in the input access matrix. Consider a maximal biclique, denote it $M_C$, which contains $C$. We know that $\tuple{u_d, p_d}\in M_C$ because $\set{u, p}\subseteq\set{u_1, \ldots, u_m, p_1, \ldots, p_n}$, which implies $\tuple{u,p}\in M_C$, and $\tuple{u_d, p_d}$ is in every maximal biclique which contains $\tuple{u, p}$. Thus, the vertices $\set{u_1, \ldots, u_m, p_1, \ldots, p_n} \cup \set{u_d, p_d}$ induce a biclique in the input access matrix. Therefore, changing $A$ to additionally assign $u_d$ and $p_d$ to $r$ leaves a sound RBAC policy with the same number of roles as $A$.
\end{proof}

For additional clarity, we point out that the proof does not go through if $\tuple{u_d, p_d}$ does not dominate $\tuple{u,p}$, because given some $u^\prime, p^\prime$ such that $u^\prime$ acquires $p^\prime$ through the role $r$ through which $u$ acquires $p$, $\set{u^\prime, u_d, p^\prime, p_d}$ may not induce a biclique in the input access matrix.

Thus, in an algorithm, we can simply remove the edge $\tuple{u_d, p_d}$ from the input, remember that it dominates $\tuple{u,p}$ as a bookkeeping exercise, and in a solution set of roles, if $u$ acquires $p$ through role $r$, simply additionally assign $u_d$ to $r$ and $r$ to $p_d$. This shrinks our input by one edge.

A next question is whether there exists an efficient algorithm to identify such dominators. Towards this, we first define the notion of \emph{adjacency} between two edges, and then a theorem that underlies an efficient algorithm. 

\begin{definition}[Adjacency]\label{def:adjacency}
We say that edges $\tuple{a, b}$ and $\tuple{c, d}$ in an access matrix $G$ are adjacent (equivalently, neighbours) if the subgraph induced by the vertices $\set{a, b, c, d}$ is a biclique.
\end{definition}

For example, the edges $\tuple{u_0, p_0}$ and $\tuple{u_2, p_0}$ are adjacent in the access matrix of Figure \ref{fig:irreducible:am}. More generally, if either $a = c$ or $b = d$, then the edges $\tuple{a, b}$ and $\tuple{c, d}$ are adjacent, assuming both those edges exist. (Of course, an edge can be said to be adjacent to itself.) As another example, the edges $\tuple{u_0, p_0}$ and $\tuple{u_2, p_1}$ are adjacent because the edges $\tuple{u_0, p_1}$ and $\tuple{u_2, p_0}$ exist as well.

\begin{theorem}\label{thm:neighbours}
Suppose $\text{neighbours}(e)$ is the set of neighbours of an edge $e$ in an access matrix. Then, an edge $d$ dominates an edge $e$ if and only if $\set{d}\cup\text{neighbours}(d) \supseteq \set{e}\cup\text{neighbours}(e)$.
\end{theorem}
\begin{proof}
For the ``only if'' direction, suppose $d$ dominates $e$ and $n \in \set{e}\cup\text{neighbours}(e)$. Then, $n$ has to be a member of some maximal biclique $C$ of $e$. We know that $d$ is a member of $C$, and therefore $C$ is a (maximal) biclique of $d$, and therefore $n \in \set{d}\cup\text{neighbours}(d)$. For the ``if'' direction, suppose $C$ is a maximal biclique of $e$, but for the purpose of contradiction, $d\not\in C$. Then, $C \subseteq \set{e}\cup\text{neighbours}(e) \subseteq \set{d}\cup\text{neighbours}(d)$. That is, every member of $C$ is a neighbour of $d$. Thus, $C \cup\set{d}$ is a biclique, thereby rendering $C$ as not maximal, our desired contradiction.
\end{proof}

\begin{figure*}
\centering
\includegraphics[width=0.75\linewidth]{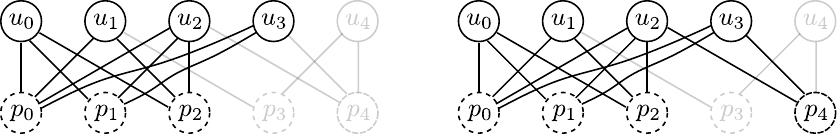}
\caption{The edge $\tuple{u_0, p_0}$ and its neighbours shown bolded to the left, and $\tuple{u_2, p_0}$ and its neighbours to the right. As the latter includes all of the former, by Theorem \ref{thm:neighbours}, $\tuple{u_2, p_0}$ dominates $\tuple{u_0, p_0}$.}
\label{fig:irreducible:neighbours}
\end{figure*}

\LinesNumbered
\DontPrintSemicolon
\begin{algorithm}
\While{not fixpoint}{
    \ForEach{edge $e$, for each neighbour $n$ of $e$}{
        \If{$n$ dominates $e$}{
            mark $n$ as removed; annotate with $e$\;
        }
    }
    \If{any edge $e$ has no neighbours}{
        remove $e$; create a new role for $e$\;
    }
}
\caption{Dominator search and edge removal}
\label{alg:dom}
\end{algorithm}

Algorithm \ref{alg:dom} is an algorithm that results from Theorem \ref{thm:neighbours}, and is exactly the prior algorithm of Ene et al.~\cite{ene08}. In Line (4), we only mark dominators for removal, and not actually remove them because otherwise, we may lose adjacency information that we will need if Algorithm \ref{alg:dom} does not remove all the edges in the input access matrix. In Line (1), by ``fixpoint'', we mean whether a dominator was indeed identified and removed in the immediately prior iteration of the \textbf{while} loop. It is possible that new dominators arise when some dominators are removed. In Line (6), we create a new role $r$ for $e = \tuple{u_e, p_e}$, i.e., for $u_e$ to acquire $p_e$ through role $r$. For any $d = \tuple{u_d, p_d}$ that dominates $e$ as identified and annotated in Line (4), we would also have $u_d$ acquire $p_d$ through $r$.

The running-time of (our version of) Algorithm \ref{alg:dom} is $O(\size{E}^4\lg{\size{E}})$ if we adopt an efficient approach for subset-checking, e.g., by treating the sets as lists, and comparing only after sorting them. The iteration of Line (1) happens $O(\size{E})$ times, and that of Line (2) $O(\size{E}^2)$ times for each iteration of Line (1). Line (3) takes time $O(\size{E}\lg\size{E})$ because we need to first generate the neighbours of $n$ and then perform the subset-check. In practice, as we show in Tables \ref{tab:dom:orig} and \ref{tab:dom:new}, our implementation runs quite fast, with progress certainly tangible.

\paragraph*{If anything remains} It is possible that all the edges in the input access matrix are marked for removal or removed as a consequence of Algorithm \ref{alg:dom}. Indeed, this is the case for all but three of the benchmark inputs of Ene et al.~\cite{ene08} --- see Table \ref{tab:dom:orig}. Otherwise, the suggestion of Ene et al.~\cite{ene08} is to adopt a reduction to clique partition, and further, reduce clique partition to graph coloring, and then deploy techniques for graph coloring. The reductions are as follows.

A clique partition in an undirected graph $G_p = \tuple{V_p, E_p}$ is a partition of the vertices $V_p$ such that the vertices in each such partition form a clique, i.e., given any pair of vertices in a partition, there is an edge between them in $E_p$. In the optimization version of the clique partition problem, we seek the fewest number of partitions; in a decision version, we ask whether there exist at most $k$ such partitions for some input integer $k$. For a reduction from biclique cover for a bipartite graph to clique partition of an undirected graph, assume the instance of the bipartite graph is $G = \tuple{V, E}$ and the output instance of clique partition is $G_p = \tuple{V_p, E_p}$. For each edge $e \in E$ we adopt a vertex $u_e \in V_p$. We introduce an edge $\tuple{u_p, v_p}\in E_p$ if and only if the edges in $G$ that correspond to $u_p$ and $v_p$ in $V_p$ are adjacent by Definition \ref{def:adjacency}. Now, there exists a biclique cover of size $k$ in $G$ if and only if there exists a clique partition of size $k$ in $G_p$.

In graph coloring we are given as input an undirected graph $G_c = \tuple{V_c, E_c}$. We seek colors (each can be perceived as an integer-label) on each vertex $u_c \in V_c$; denote such as color as $l(u_c)$. We require the fewest different number of colors (i.e., size of the range of the coloring-function $l$) such that $\tuple{u_c, v_c}\in E_c$ implies $l(u_c)\not= l(v_c)$, i.e., adjacent vertices get different colors. In a decision version, we are additionally given an integer $k$ and we ask whether at most $k$ colors suffice for $G_c$. In a reduction from clique partition, given an instance $G_p = \tuple{V_p, E_p}$ to graph coloring, an instance of which is $G_c = \tuple{V_c, E_c}$, we adopt $V_c = V_p$, and as $E_c$, we adopt the complement of $E_p$, i.e., $\tuple{u, v}\in E_c$ if and only if $\tuple{u,v}\not\in E_p$. Now, there exist at most $k$ clique partitions in $G_p$ if and only if $G_c$ can be colored with at most $k$ colors.

\paragraph*{(In)effectiveness} There are two ways to measure the goodness of Algorithm \ref{alg:dom}. One is time-performance, i.e., how long it takes to run benchmark inputs. Another, and perhaps more important, is how much of the input access matrix remains after. We present our results for these two measures in Tables \ref{tab:dom:orig} and \ref{tab:dom:new}.

\begin{table*}
\caption{The performance of Algorithm \ref{alg:dom} on the original benchmark of Ene et al.~\cite{ene08}. Our results are for our implementation \cite{mycode}, and coincide with those reported in that work.}
\label{tab:dom:orig}
\begin{tabular}{cccccc}
\toprule
{Instance} & 
{\# edges} & 
{\# edges after Alg.~1} & 
{\% edges after Alg.~1} &
{Time (min:sec)}\\
\midrule
al & $185,294$ & $97$ & $0.05$ & $11$:$49$\\
apj & $6,841$ & $0$ & $0$ & $< 1\text{s}$\\
as & $105,205$ & $44$ & $0.04$ & $9$:$35$\\
domino & $730$ & $0$ & $0$ & $< 1\text{s}$\\
emea & $730$ & $0$ & $0$ & $< 1\text{s}$\\
fw1 & $31,951$ & $0$ & $0$ & $1$:$11$\\
fw2 & $36,428$ & $0$ & $0$ & $1$:$52$\\
hc & $1486$ & $0$ & $0$ & $< 1\text{s}$\\
mailer & $5552$ & $22$ & $0.4$ & $< 1\text{s}$\\
univ & $3954$ & $0$ & $0$ & $< 1\text{s}$\\
\bottomrule
\end{tabular}
\end{table*}

\begin{table*}
\caption{The performance of Algorithm \ref{alg:dom} on the new benchmark \cite{anderer21}.}
\label{tab:dom:new}
\begin{tabular}{cccccc}
\toprule
{Instance} & 
{\# edges} & 
{\# edges after Alg.~1} & 
{\% edges after Alg.~1} &
{Time (m:s)}\\
\midrule
small 01 & $600$ & $183$ & $30.5$ &  $< 1\text{s}$\\
small 02 & $1082$ & $501$ & $46.3$ &  $< 1\text{s}$\\
small 03 & $1369$ & $0$ & $0$ &  $< 1\text{s}$\\
small 04 & $1932$ & $736$ & $38.1$ &  $< 1\text{s}$\\
small 05 & $1372$ & $0$ & $0$ &  $< 1\text{s}$\\
small 06 & $2152$ & $1044$ & $48.5$ &  $< 1\text{s}$\\
small 07 & $9371$ & $2603$ & $27.8$ & $0$:$25$ \\
small 08 & $4415$ & $1538$ & $34.8$ & $0$:$02$  \\
medium 01 & $15,567$ & $3724$ & $23.9$ & $0$:$08$\\
medium 02 & $33,959$ & $17,855$ & $52.6$ & $1$:$05$\\
medium 03 & $22,988$ & $11,855$ & $51.6$ & $0$:$18$\\
medium 04 & $23,949$ & $4322$ & $52.6$ & $0$:$09$\\
medium 05 & $47,674$ & $27,660$ & $58$ & $1$:$21$\\
medium 06 & $48,058$ & $21,954$ & $45.7$ & $1$:$07$\\
large 01 & $60,288$ & $27,092$ & $44.9$ & $1$:$29$\\
large 02 & $49,579$ & $43,379$ & $87.5$ & $1$:$10$\\
large 03 & $23,778$ & $14,000$ & $58.9$ & $0$:$10$\\
large 04 & $74,347$ & $4097$ & $5.5$ & $0$:$16$\\
large 05 & $148,067$ & $10,242$ & $6.9$ & $1$:$19$\\
large 06 & $62,292$ & $4459$ & $7.2$ & $0$:$11$\\
comp 01.1 & $49,283$ & $10,399$ & $21.1$ & $0$:$55$\\
comp 01.2 & $60,564$ & $14,563$ & $24$ & $1$:$14$\\
comp 01.3 & $56,981$ & $10,907$ & $19.1$ & $1$:$16$\\
comp 01.4 & $70,122$ & $16,527$ & $23.6$ & $1$:$42$\\
comp 02.1 & $278,809$ & $34,029$ & $12.2$ & $1$:$20$:$38$\\
comp 02.2 & $332,985$ & $82,035$ & $24.6$ & $56$:$39$\\
comp 02.3 & $411,702$ & $98,656$ & $24$ & $2$:$24$:$59$\\
comp 02.4 & $494,455$ & $145,801$ & $29.5$ & $2$:$50$:$59$\\
comp 03.1 & $646,406$ & $226,787$ & $35$ & $14$:$45$:$24$\\
comp 03.2 & $812,337$ & $377,957$ & $46.5$ & $12$:$09$:$52$\\
comp 03.3 & $1,053,246$ & $495,728$ & $47$ & $35$:$10$:$20$\\
comp 03.4 & $1,331,702$ & $762,248$ & $57.2$ & $42$:$40$:$17$\\
comp 04.1 & $573,065$ & $113,474$ & $19.8$ & $1$:$20$:$47$\\
comp 04.2 & $665,217$ & $160,180$ & $24.1$ & $1$:$03$:$14$\\
comp 04.3 & $726,448$ & $138,355$ & $19.1$ & $1$:$56$:$55$\\
comp 04.4 & $844,106$ & $183,709$ & $21.8$ & $1$:$51$:$24$\\
rw 01 & $383,216$ & $0$ & $0$ & $8$:$18$\\
\bottomrule
\end{tabular}
\end{table*}

From the standpoint of time-performance, (our implementation of) Algorithm \ref{alg:dom} performs well, and as such confirms the findings of Ene et al.~\cite{ene08} for their benchmark inputs. Given that fewer than 100 edges remain for any of their inputs, for any edges that remain, a reduction from a decision version of biclique cover to  binary Linear Programming (LP), and then use of a corresponding solver would work. No reduction to clique partition, and then further reduction to coloring, as proposed by Ene et al.~\cite{ene08} is needed. Such a reduction is to adopt a target value $k$ for the minimum size of a biclique cover, and carry out binary search on $k$. We begin the binary-search with lower-bound $1$, as the graph has at least $1$ edge, and an upper-bound of $\min\set{\size{\text{users}}, \size{\text{permissions}}, \size{\text{edges}}}$ --- we know that there exists an RBAC policy with at most that many roles. We adopt a binary variable $x_{a,b}$ for each edge $a$ and each $b\in\set{1,\ldots, k}$, where $x_{a,b} = 1$ if edge $a$ is in biclique $b$, and $0$ otherwise. Our linear program is:
\begin{align*}
&\displaystyle\sum_{b=1}^k x_{a,b} \ge 1 &\forall \text{ edges } a\\
& x_{a,b} + x_{c,b} \le 1 &\forall b\in\set{1,\ldots, k}, \forall~a, c\text{ not adjacent}
\end{align*}

The first constraint ensures that every edge is in at least one biclique, and the second ensures that two edges are in the same biclique only if they are adjacent. Table \ref{tab:bcbs} shows the time-performance for the three instances of the original set of benchmarks, and three newer benchmark inputs which have the fewest, $> 0$, of remaining edges after Algorithm \ref{alg:dom}. 

\begin{table*}
\caption{The time-performance of the reduction to LP from a decision version of biclique cover + binary search on three benchmark inputs from each of the original and newer sets for inputs that have the fewest $> 0$ edges that remain after Algorithm \ref{alg:dom}. The double lines separate the original from the new benchmarks.}
\label{tab:bcbs}
\centering
\begin{tabular}{ c  c || c  c }
\toprule
{Instance} & 
{Time (m:s)} &
{Instance} & 
{Time (m:s)}\\
\midrule
al & $36$:$05$ & small 01 & $10$:$36$\\
as & $0$:$02$ & small 02 & $> 1\text{ day}$\\
mailer & $< 1\text{s}$ & small 04 & $> 1\text{ day}$\\
\bottomrule
\end{tabular}
\end{table*}

Among the newer benchmarks, even though this approach scales for small 01, which has a modest number of remaining edges, for the next smallest inputs, which are small 02 and small 04, even a day does not suffice. Indeed, in the case of the latter, a day does not suffice for even the first iteration of the binary search. Also, there is lack of tangible progress unless the binary search progresses: we do not know whether the constraint solver is ``stuck'' \cite{solverstuck}, or it will return if we wait longer.

As for Algorithm \ref{alg:dom}, on the newer benchmark inputs, it is effective on few inputs only. As Table \ref{tab:dom:new} shows, unlike with the earlier benchmarks (Table \ref{tab:dom:orig}), a significant portion of the edges remain for all but 3 of the 37 inputs. More than the proportion of edges that remain, is the raw size. The reduction to clique partition can result in much larger graphs: 10$\times$ to 1000$\times$ in the size of the bipartite graph. In Table \ref{tab:cp:blowup}, we show this blowup for some of the newer benchmarks. A further reduction to graph coloring exacerbates this problem as the graphs after reduction to clique partition are relatively sparse. The further reduction to coloring computes the complement of the graph. If we consider, for example, the benchmark small 07, its maximum possible number of edges is $\binom{2603}{2} \approx 3 \text{ million}$, because $2603$ edges remain after Algorithm \ref{alg:dom} (see Table \ref{tab:dom:new}), each of which becomes a vertex in the graph that is the instance of clique partition. Of these, we have only about $1 \text{ million}$ edges (see Table \ref{tab:cp:blowup}), which means that the complement is roughly double in size. It is much worse with, for example, the benchmark inputs large 02 and comp 04.3; for those, as Table \ref{tab:cp:blowup} indicates, the instance of graph coloring is 10,000$\times$ the size of the bipartite graph after Algorithm \ref{alg:dom}.

\paragraph*{Takeaways} (1) Algorithm \ref{alg:dom} is a useful first step; it can trim the size of the access matrix we need to consider next. (2) Reduction to clique partition and then graph coloring does not scale for several of the benchmark inputs. The reduction to LP from biclique cover with binary search does not perform well either, as we see from Table \ref{tab:bcbs}. Therefore, we need a different approach, which is exactly our main contribution. In our approach, we remain in the realm of bipartite graphs and bicliques to mitigate the blowup in size.

\noindent
\begin{table*}
\caption{The size of the access matrix after Algorithm \ref{alg:dom}, and of the graph after reduction to clique partition, and then to coloring, for some benchmark inputs. Each size is the sum of the number of vertices and edges. The double lines separate the original from the new benchmarks.}
\label{tab:cp:blowup}
\centering
\begin{tabular}{cccc}
\toprule
{\multirow{2}{*}{Instance}} & 
{Size after} &
{Size after reduction} &
{Size after reduction}\\
& Algorithm~\ref{alg:dom} & 
to clique partition &
to coloring (approx.) \\
\midrule
small 07 & $2895$ & $1,014,976$ & $2$ million\\
medium 05 & $29,139$ & $10,708,568$ & $400$ million\\
large 02 & $45,370$ & $8,448,482$ & $900$ million\\
comp 01.3 & $13,642$ & $1,053,419$ & $58$ million\\
comp 04.3 & $162,775$ & $64,055,501$ & $9.5$ billion \\
\bottomrule
\end{tabular}
\end{table*}

\section{Our Technique: Enumerate Maximal Bicliques}
\label{sec:enum}

\begin{figure*}
\centering
\includegraphics[width=0.85\linewidth]{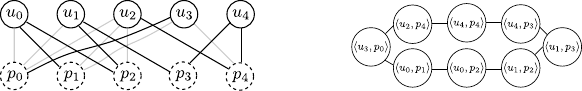}
\caption{The graph to the left shows in bold the edges that remain in the graph from Figure \ref{fig:irreducible:am} after Algorithm \ref{alg:dom}
\label{fig:irreducible:postdom}. The graph to the right helps visualize the maximal bicliques; it is the graph after reduction to clique partition.}
\end{figure*}

We now describe our approach. First, we add a constraint to the problem we consider. Rather than requiring only that the roles correspond to a minimum number of bicliques, we require each such role to correspond to a maximal biclique. This does not impact soundness, as the following theorem states.

\begin{theorem}\label{thm:maxcliques}
Suppose there exists an RBAC policy of $k$ roles for an input access matrix. Then there exists an RBAC policy of $k$ roles for that input, where each role corresponds to a maximal biclique.
\end{theorem}
\begin{proof}
By construction. Suppose an RBAC policy of $k$ roles is $A$. We know that each role $r$ corresponds to a biclique. Let the users assigned to $r$ be $\set{u_1, \ldots, u_m}$ and the permissions to which $r$ is assigned be $\set{p_1, \ldots, p_n}$. In the input, let $C$ be the biclique induced by $\set{u_1, \ldots, u_m, p_1, \ldots, p_n}$. If $C$ is maximal, there is nothing to be done. Otherwise, let $M_C$ be a maximal biclique that includes $C$, and let the set of users and permissions in $M_C$ be $\set{u_1, \ldots, u_m, u_{m+1}, \ldots, u_{m+a}, p_1, \ldots, p_n, p_{n+1}, \ldots, p_{n+b}}$. Simply change $A$ to additionally assign the users $u_{m+1}, \ldots, u_{m+a}$ to $r$, and $r$ to the permissions $p_{n+1}, \ldots, p_{n+b}$.
\end{proof}

Limiting ourselves to maximal bicliques decreases the number of possibilities we need to consider: there are fewer maximal bicliques than bicliques. From the standpoint of worst-case computational hardness, however, this makes no difference, at least under a polynomial-time Turing reduction \cite{hopcroft2001introduction}. That is, if the original problem is in \cp, so is the new one: given a solution for the original problem, we can simply test each edge for addition to a biclique to make it maximal. Of course, the converse is true as well.

We then amend the problem even further. Given the access matrix and edges marked for removal after Algorithm \ref{alg:dom}, we generate all maximal bicliques, denote the set $M$. We then address the problem of finding the minimum possible number of maximal bicliques from $M$ we need; these would correspond exactly to the minimum number of roles in the access matrix that remains after Algorithm \ref{alg:dom}. Thus, our problem after we generate the set $M$ is: given as input $\tuple{B, E_{\text{rem}}, M}$, where $B = \tuple{V_B, E_B}$ is an access matrix, $E_{\text{rem}}\subseteq E_B$ is a set of edges that have been marked for removal, and $M$ is the set of all maximal bicliques in $B$ that includes only edges from $E_B\setminus E_{\text{rem}}$ (but $E_{\text{rem}}$ is included in determining adjacency of edges in $E_B\setminus E_{\text{rem}}$), what is a minimum-sized set of roles?

\paragraph*{Is this new problem in \cp?} Towards an algorithm, this is a natural question that arises. We conjecture that the answer is `no', but do not know definitively. A basis for our conjecture is that when the minimum size of a biclique cover in a bipartite graph is small compared to the size of the instance, the problem remains \cnph; specifically, if the minimum number of bicliques in a biclique cover is $O(\lg n)$, where $n$ is the size of the input, the problem remains $\cnph$ \cite{gott05}. As another, admittedly rough, intuition, suppose the size of the input access matrix is $n$, the number of maximal bicliques in it is $\approx n$, and the minimum size of a biclique cover is $\approx \sqrt{n}$. Then, to identify a biclique cover from the input set of all maximal bicliques, a brute-force algorithm would need to check $\approx \binom{n}{\sqrt{n}}$ possibilities, which is superpolynomial in $n$.

Consequently, rather than attempting to devise an efficient algorithm, we rely on a constraint solver. Once we have computed $M$, encoded as a set of edges, from $B$ and $E_{\text{rem}}$, we adopt the following reduction to Integer Linear Programming (ILP). For each $c\in M$, we adopt a binary variable $x_c$, with the semantics that $x_c = 1$ if and only if the maximal biclique $c$ is in our solution set of bicliques. Our linear program is:

\begin{align}
\text{minimize} & \displaystyle\sum_{c\in M} x_c \nonumber\\
 \label{lp:maxsets}
 \\
\text{subject to} & \displaystyle\sum_{\set{d \in M : e \in d}} x_d \ge 1 & \forall e\in E_B \nonumber
\end{align}

The first line ensures that we minimize the number of maximal bicliques in our solution.  The second ensures that at least one maximal biclique for each edge appears in a solution. This approach is a specialization of the branch-and-price technique \cite{mehrotra96,barnhart98}, which we discuss below. So we can refer to it in other parts of this paper, we capture the above as Algorithm \ref{alg:maxsets}.

\LinesNumbered
\DontPrintSemicolon
\begin{algorithm}
Generate the set of all maximal bicliques, $M$, from $B$ and $E_{\text{rem}}$\;
Reduce to ILP (\ref{lp:maxsets})\;
Invoke constraint-solver\;
\caption{Algorithm that leverages maximal bicliques}
\label{alg:maxsets}
\end{algorithm}

Consider, for example, the access matrix from Figure \ref{fig:irreducible:am}, after we run Algorithm \ref{alg:dom} on it. The edges that remain are shown in bold to the left of Figure \ref{fig:irreducible:postdom}. To the right is a graph that helps us deduce visually the maximal bicliques in it: we adopt as vertices in this graph the edges that remain in the bipartite graph, and there is an edge between two vertices if and only if the edges corresponding to those vertices are adjacent under Definition \ref{def:adjacency} in the bipartite graph. (This is exactly the reduction to clique partition of Ene et al.~\cite{ene08}.)

The graph to the right tells us that we have two maximal bicliques per edge, each of size two, for a total of eight maximal bicliques: $\set{\tuple{u_3, p_0}, \tuple{u_2, p_4}}, \ldots, \set{\tuple{u_0, p_1}, \tuple{u_3, p_0}}$. This is our set $M$ of maximal bicliques. Assume we designate these sets $0, \ldots, 7$. Then the constraints for the edges are $x_0 + x_1 \ge 1$, $x_1 + x_2 \ge 1$, \ldots, $x_7 + x_0 \ge 1$. A solution is the four maximal bicliques $\set{\tuple{u_0, p_1}, \tuple{u_3, p_0}}$, $\set{\tuple{u_2, p_4}, \tuple{u_4, p_4}}$, $\set{\tuple{u_4, p_3}, \tuple{u_1, p_3}}$ and $\set{\tuple{u_1, p_2}, \tuple{u_0, p_2}}$. Once we create a role for each of these four maximal bicliques and assign also the users and permissions for edges that dominate as identified in Line (4) of Algorithm \ref{alg:dom}, we get the role-based policy to the right of Figure \ref{fig:irreducible:rbac}.

\paragraph*{Measure of hardness of an instance} We observe that from the original problem, i.e., given a bipartite graph, determine (the cardinality of) a minimum-sized biclique cover, such an instance of ILP is not polynomial-time computable, because the number of maximal bicliques may be exponential in the size of the bipartite graph. However, if an input does not have a large number of maximal bicliques, this approach is efficient in practice. Indeed, we argue that a meaningful measure of the hardness of an instance corresponds exactly to whether we can enumerate all maximal bicliques quickly. This is a function of the total number of maximal bicliques.

\paragraph{Branch-and-price} As we say above, our proposed approach of enumerating all maximal bicliques and then choosing from amongst them is a specialization of the branch-and-price technique that has been employed for other problems that are \cnph \cite{mehrotra96,barnhart98}. We include the corresponding code for our problem as part of our open-source offering \cite{mycode}. The approach is to start with an initial set of maximal bicliques that is not necessarily the entire set as in Algorithm \ref{alg:maxsets}, but rather a subset. We set the LP variables $x_c$ in LP (\ref{lp:maxsets}) to be real values, rather than binary. We then solve the LP with that subset of maximal bicliques. Once we obtain the objective value, we retrieve also the dual, or $\pi$, value of each constraint in the ``subject to'' of ILP (\ref{lp:maxsets}). Note that that set of constraints comprises exactly one constraint for each edge in our access matrix. Thus, we have a $\pi_e$ value corresponding to each edge $e$.  We then solve the following linear program, which weighs each variable with the corresponding $\pi$ values. (Denote as $E_B$ our set of edges in the access matrix.)

\begin{align}
\text{maximize} && \displaystyle\sum_{e\in E_B} \pi_e\cdot x_e \nonumber\\
 \label{lp:dual}
 \\
\text{subject to } && x_e + x_f \le 1 && \forall e, f \in E_B \text{ not adjacent} \nonumber
\end{align}

In this linear program LP (\ref{lp:dual}), we set the $x_e$'s to be binary. If the optimal (maximum) value is $> 1$, then we identify those $x_e$'s whose value is $1$ in the solution, and add maximal bicliques for those edges to our working set of maximal bicliques in the next iteration. If the maximum objective value is $\le 1$, then we have converged on a final solution to our original LP (\ref{lp:maxsets}). Mehrotra and Trick \cite{mehrotra96} propose exactly such an approach for graph coloring. Our algorithm is as shown in Algorithm \ref{alg:dual}.

\LinesNumbered
\DontPrintSemicolon
\begin{algorithm}
$\textit{done} \gets \textsf{False}$\;
$M\gets \text{initial set of maximal bicliques}$\;
\While{not \textit{done}}{
    Construct and solve LP (\ref{lp:maxsets}) with $M$, each $x_c$ real\;
    Collect the set of $\pi_e$ values for each $e\in E_B$\;
    Construct and solve LP (\ref{lp:dual}) with each $x_e$ binary\;
    \lIf{objective value $\le 1$}{
        $\textit{done}\gets \textsf{True}$
    }
    \Else{
        Identify $x_e$'s whose value is $1$\;
        Generate more maximal bicliques for those edges, $e$\;
        Add to $M$
    }
}
\caption{Branch-and-price algorithm that leverages duals}
\label{alg:dual}
\end{algorithm}

Some downsides are immediately apparent with this approach. If the set of all maximal bicliques is sufficiently small given our computing resources, we might as well simply employ Algorithm \ref{alg:maxsets}. If the set of all maximal bicliques is large, as is the case with our harder benchmark inputs (see Table \ref{tab:hard}), then our working set of maximal bicliques could grow quite large, to a point that our computing resources no longer suffice to build the LP models, and the run the constraint solver.

Furthermore, as our experience with the benchmark input small 07, which has more than 45 million maximal bicliques after applying Algorithm \ref{alg:dom}, suggests, this approach does not always show tangible progress. This is because we do not know whether the maximal bicliques we generate additionally and add to M in Lines (10)-(11) of Algorithm \ref{alg:dual} necessarily bring us any closer to converging to the optimal value of LP (\ref{lp:maxsets}). Figure \ref{fig:mwplot} shows  the covergence of this approach for small 07, after Algorithm \ref{alg:dom}, over almost 5 days. As we see, the initial objective value is $109$, which is quite far from the known upper-bound of $30$. There are long periods during which there is very little, or no tangible progress. The final value after the time we ran this was $43.83$; still quite far from the known upper-bound. Thus, the branch-and-price approach appears to be ineffective for the benchmarks for role mining.

\begin{figure}
\centering
\includegraphics[width=0.67\linewidth]{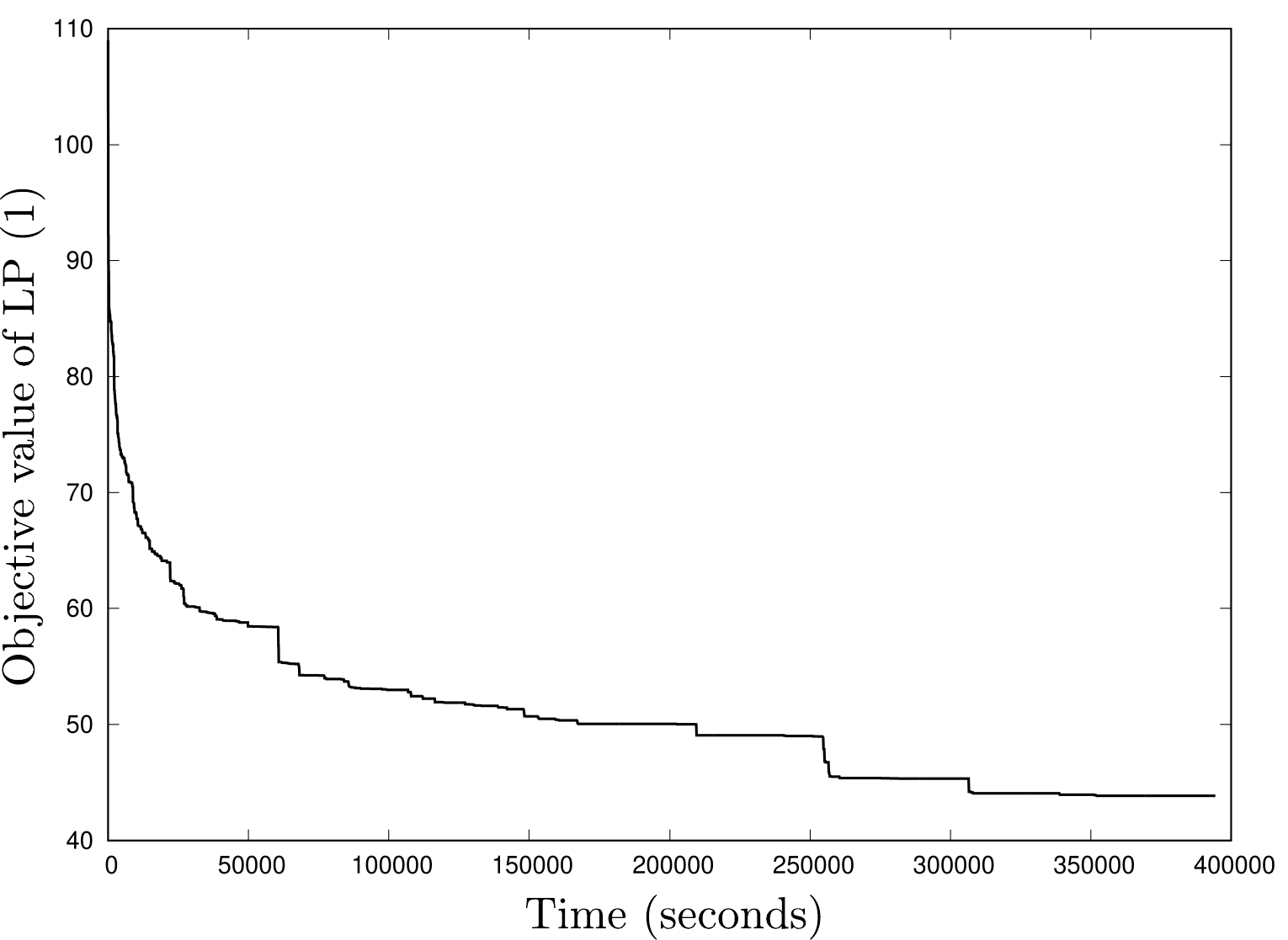}
\caption{The convergence of the objective value of LP (\ref{lp:maxsets}) for branch-and-price, i.e., Algorithm \ref{alg:dual}, over about 5 days for the benchmark small 07 after Algorithm \ref{alg:dom}.}
\label{fig:mwplot}
\end{figure}

\paragraph*{To compute $M$} Before we discuss the performance of Algorithm \ref{alg:maxsets}, we address the issue of enumerating maximal bicliques in our access matrix. We observe that this problem remains an area of active research (see, for example, \cite{zhang14,abidi20,lu22}).  The algorithm we use is based on that of Cazals and Karande \cite{cazals08} for enumerating maximal cliques in an undirected graph, which itself, as they discuss in that work, is a composition of three earlier pieces of work. More specifically, our implementation adapts the implementation of the routine \textsf{find\_cliques} of the \textsf{networkx} software package \cite{networkx}. Of course, any algorithm to generate bicliques in a bipartite graph efficiently would work for us; see our discussion about implementation considerations below, however.

An initial version of the algorithm is to maintain three disjoint sets of nodes, $R, P, X$ where $R$ is a clique so far, $P$ are nodes that are neighbours of all nodes in $R$ and therefore any may be added to $R$ while maintaining $R$'s clique-ness, and $X$ are nodes for each of which all maximal cliques have already been computed. If $P$ and $X$ are empty, we immediately return $R$ as a maximal clique (or add it to a set of maximal cliques to be returned --- see our discussion on implementation below). Otherwise we pick a node, call it $u\in P$, add it to $R$, update $P$ as $\left(P\setminus\set{u}\right) \cap N[u]$, where $N[u]$ are all the neighbours of $u$, and update $X$ as $X \cap N[u]$. Once we recursively call the algorithm with the new $R, P$ and $X$, we update $X$ as $X\cup \set{u}$. (Initially, the algorithm is invoked with $R = X = \emptyset$ and $P$ as all the vertices.) This basic algorithm is then improved to make the choice of $u$, which is called a pivot, carefully to improve efficiency --- see Cazals and Karande \cite{cazals08}.

Our choice of algorithm for maximal biclique generation is based largely on implementation considerations. We have experimented with the implementations of Zhang et al.~\cite{zhang14} and Lu et al.~\cite{lu22}, in addition to the one we adopted to then adapt for our purposes. Our choice came down to a few, somewhat minor preferences. One is that \textsf{networkx} is implemented in Python, which makes it easier for us to interface with the rest of our code, which is written in Python as well. Another, perhaps more important, is that the \textsf{networkx} implementation is a generator function: it does not immediately compute its entire set of results, but rather, returns a kind of iterator object for one maximal biclique at a time. In Algorithm \ref{alg:maxsets}, we first check whether the number of maximal bicliques is indeed upper-bounded by a threshold we set --- 3 million, in our experiments. This check can be done with high memory efficiency if we can ``throw away'' each maximal biclique after incorporating it in our count, which is eased by \textsf{networkx}'s choice of implementing as a generator. The adaptation of the \textsf{networkx} implementation required two things: changing what ``adjacent'' means to match Definition \ref{def:adjacency}, and incorporating the edges marked for removal $E_{\text{rem}}$. Both these turned out to be easy as well with that implementation. We include our adapted version of the routine as part of our open-source offering, of course with credit to the source \cite{mycode}.

\paragraph*{Performance of Algorithm \ref{alg:maxsets}} Table \ref{tab:maxsets} shows the performance of Algorithm \ref{alg:maxsets} on benchmark inputs that have $> 0$ edges after Algorithm \ref{alg:dom}, and that have a number of maximal bicliques that is smaller than a threshold we set. Our threshold is based on trial-and-error with our setup against the benchmarks; we have set it to 3 million maximal bicliques. As for the benchmark inputs not in the table, each easily has more than 3 million maximal bicliques. For example, we ran our routine to enumerate maximal bicliques over longer than $12$ hours to count the number in small 07, which has seemingly only a modest number of edges, $2603$, after Algorithm \ref{alg:dom}. And it has more than 45 million maximal bicliques. Our system certainly does not have the resources to handle a reduction to ILP (\ref{lp:maxsets}) for such an instance. From the standpoint of our experimental setup and system, we deem any instance that has more than 3 million maximal bicliques after Algorithm \ref{alg:dom} as hard.

The ``LP (\ref{lp:maxsets}) build + solve time [h:]m:s'' includes both the time to build the LP model and constraints for the constraint solver, and the time the solver takes once we issue ``\textsf{solve()}'' to it. Our data confirms that the number of maximal bicliques is indeed a good measure of hardness. The Pearson correlation coefficient, which measures the linear correlation between two sets of data-values, between the ``\# maximal bicliques'' column and the ``LP (\ref{lp:maxsets}) build + solve time [h:]m:s'' column, is 0.9132, which is a high correlation. Indeed, the size of the instance as measured by the size of the access matrix is certainly not a good measure of hardness of the instance. For example, the instance large 04 has more than 74,000 edges (see Table \ref{tab:dom:new}) while medium 06 has only around 48,000 edges. Yet the number of maximal bicliques in the latter is more than 2 million, whereas the number of maximal bicliques in the former is fewer than 2000 only. And this correlates to the stark difference in how hard they are as we report in Table \ref{tab:maxsets}: medium 06 takes longer than 5 hours, whereas large 06 takes a second only.

The last three columns are each a number of roles: the number computed by the constraint solver for LP (\ref{lp:maxsets}), ``\# roles, LP'', the total number of roles, i.e., including those computed earlier in Line (6) of Algorithm \ref{alg:dom}, ``\# roles total: Alg.~\ref{alg:dom} Line (6) + LP'', and any known exact or upper-bound for the number of roles against which we can check, ``\# roles, known bound''. We have an exact bound from the work of Ene et al.~\cite{ene08} for the al and as benchmarks, but have nothing for mailer, which was not part of their work and was presumably added later only. For the newer benchmarks of Anderer et al.~\cite{anderer21} that we show in the table, as we mention in Section \ref{sec:intro}, we have a known upper-bound for the number of roles. We observe that in all but three cases in that table, the minimum number of roles is exactly the known upper-bound. In three cases, i.e., small 01, medium 03 and large 03, the minimum is one role fewer. Note also from Table \ref{tab:dom:nroles} that small 05's minimum number of roles is 49, which is one fewer than the known upper-bound of 50. Thus, it appears that the RBAC policies from which the access matrix instances in the new benchmark were derived are quite ``tight'' to begin with, in that they are either exactly, or very close, to the fewest number of roles possible.

\begin{table*}
\caption{The number of roles as determined in Line (6) for benchmark inputs that Algorithm \ref{alg:dom} solves fully. In all cases, we match or beat a known exact/upper-bound. For the original benchmark inputs, our results match those of Ene et al.~\cite{ene08}. The double lines separate the original from the new benchmarks.}
\label{tab:dom:nroles}
\centering
\begin{tabular}{ c  c || c  c }
\toprule
{Instance} & 
{\# roles} &
{Instance} & 
{\# roles}\\
\midrule
apj & 453 & small 03 & 25 \\
domino & 20 & small 05 & 49\\
emea & 34 & rw 01 & 463\\
fw1 & 64 \\
fw2 & 10 \\
hc & 14\\
univ & 18\\
\hline
\end{tabular}
\end{table*}

In summary, the sequential composition of Algorithm \ref{alg:dom} and our technique of enumerating all maximal bicliques if any edges remain, the reduction to ILP (\ref{lp:maxsets}), and use of an ILP solver, works well in cases that we have only as many maximal bicliques as a threshold we set based on our computing resources. All the original benchmarks are addressed efficiently by this sequential composition. Of the newer benchmarks, 22 of the 37 are addressed. We point out also that it is easy to check if this sequential composition will work with tangible progress: we simply run Algorithm \ref{alg:dom}, which certainly demonstrates tangible progress, and then count the number of maximal bicliques using our implementation for enumerating them while checking against the threshold.

\noindent
\begin{table*}
\caption{The performance of our Algorithm \ref{alg:maxsets} that enumerates all maximal bicliques and then runs a constraint solver for the ILP (\ref{lp:maxsets}), for the benchmark inputs each of which has $> 0$ edges after Algorithm \ref{alg:dom}, and has $<$ 3 million maximal bicliques. The double lines separate the original from the new benchmarks.}
\label{tab:maxsets}
\centering
\begin{tabular}{cccccc}
\toprule
Instance (after & \# maximal & LP (\ref{lp:maxsets}) build + & \multirow{2}{*}{\# roles, LP} & \# roles total: Alg. & \# roles, known\\
Alg.~\ref{alg:dom}) & bicliques & solve time [h:]m:s & & \ref{alg:dom} Line (6) + LP & bound\\
\midrule
al & 114 & $<$ 1s & 27 &  398 & 398 \\
as & 61 & $<$ 1s & 19 & 178 & 178 \\
mailer & 22 & $<$ 1s & 11 & 565 & \emph{(none)} \\
\midrule
\midrule
small 01 & 449 & $<$ 1s & 20 & 24 & 25\\
small 02 & 20,800 & 0:05 & 24 & 25 & 25 \\
small 04 & 50,417 & 0:20 & 25 & 25 & 25\\
small 06 & 10,056 & 0:06 & 47 & 50 & 50\\
small 08 & 85,901 & 0:51 & 47 & 50 & 50\\
medium 01 & 15,383 & 0:12 & 92 & 150 & 150\\
medium 03 & 503,388 & 37:51 & 199 & 199 & 200 \\
medium 04 & 10,696 & 0:09 & 181 & 200 & 200\\
medium 06 & 2,325,223 & 5:01:45 & 250 & 250 & 250\\
large 01 & 726,965 & 2:16:23 & 245 & 250 & 250\\
large 02 & 664,168 & 3:55:53 & 499 & 500 & 500\\
large 03 & 34,647 & 3:06 & 466 & 499 & 500\\
large 04 & 1823 & 0:01 & 331 & 400 & 400\\
large 05 & 12,442 & 0:21 & 400 & 400 & 400\\
large 06 & 1869 & 0:01 & 442 & 500 & 500\\
comp 01.1 & 98,596 & 6:37 & 361 & 400 & 400\\
comp 01.2 & 171,043 & 17:48 & 388 & 400 & 400\\
comp 01.3 & 132,219 & 11:12 & 365 & 400 & 400\\
comp 01.4 & 288,401 & 35:08 & 389 & 400 & 400\\
\bottomrule
\end{tabular}
\end{table*}

\section{Addressing Hard Instances}\label{sec:hard}

We are left with 15 instances that are hard under our characterization under ``Measure of hardness of an instance'' in the previous section. Given the underlying computational hardness of the problem (see Section \ref{sec:intro}), we are left to resort to heuristics. For a heuristic, we adopt three axes for trade-offs. (1) Optimality --- this is our quality objective; we seek that the number of roles is minimized. This is what we primarily want. (2) Time-efficiency in practice --- we seek an approach that runs fast in a practical sense. We must appreciate that if an instance has a large number of maximal cliques, then we need to adjust our time-efficiency expectations accordingly. (3) Tangible progress --- we seek an approach with which progress is tangible. Thus, we do not blindly outsource to an opaque constraint-solver which gives us no meaningful indication of forward progress. These are axes for trade-offs, i.e., we adopt the mindset that for hard instances, we may not always be able to attain (1), optimality. However, if that is the case, we expect to gain along at least one of the other axes, (2) and/or (3). In particular, we consider (3), tangible progress, to be ``non-negotiable'', i.e., we are willing to trade-off along (1) and/or (2) to achieve (3).

With these trade-offs in mind, we first revisit a prior heuristic, also due to Ene et al.~\cite{ene08}. This prior heuristic gives us a baseline against which to compare our new heuristic. We include an implementation of that heuristic in our open-source offering \cite{mycode}.

\paragraph*{A prior heuristic} The prior heuristic \cite{ene08} is a sequential composition of two algorithms: the computation of an initial set of bicliques (i.e., roles) via a greedy technique, and then what that works calls a lattice-based postprocessing. These are as follows.

Towards computing an initial set of roles, in the input access matrix, we first pick a vertex, i.e., a user or permission. Our basis for the choice of a user/permission yields two versions of this heuristic. We pick one that either (i) has the smallest degree, i.e., fewest edges incident on it, or, (ii) the largest degree. Once we pick such a user/permission, denote it $v$, we compute the set $S$ of all its neighbours. If $v$ is a user, $S$ comprises all permissions which $v$ possesses; if $v$ is a permission, $S$ comprises all users who are authorized to $v$. Then, we pick all the vertices, denote the set $T$, each of which has an edge to every member of $S$. Thus, if $v$ is a user, $T$ comprises users; otherwise $T$ comprises permissions. We have now computed the biclique, $C = T\times S$. We know that $C\not= \emptyset$ because the user/permission with which we started, $v\in T$, and $v$ has degree at least one. We adopt $C$ as a role, remove all its edges, and repeat.

In the access matrix in Figure \ref{fig:irreducible:am} in Section \ref{sec:intro}, for example, if choose a vertex of smallest degree, i.e., strategy (i) above, we would pick one of $u_4$ or $p_3$, each of which has degree two. Suppose we pick $p_3$. Then, $S = \set{u_1, u_4}, T = \set{p_3}$. If we adopt strategy (ii), we would pick $v = u_2$, and $S = \set{p_0, p_1, p_2, p_4}, T = \set{u_2}$. We run the algorithm for each of strategies (i) and (ii), and pick the one that yields the fewest roles after the lattice-based postprocessing step below. Table \ref{tab:eneheuristic} shows the results for the hard instances. Note the number of roles can vary across repeated runs of the algorithm depending on the particular user/permission of smallest/largest degree we happen to pick in an iteration.

The second step is what Ene et al.~\cite{ene08} call a lattice-based postprocessing of the roles we compute in the first step above. That work in turn credits the work of Zhang et al.~\cite{zhang07} for this algorithm. The intent is to reduce the number of roles computed by the greedy algorithm above, or really, any prior algorithm that determines sets of roles. Use $\mathsf{perms(\cdot)}$ to denote the function that maps a role to the set of permissions to which it is assigned. The algorithm is to identify distinct roles $r, r_1, \ldots, r_k$ with the following property: $\mathsf{perms}(r) \supseteq \displaystyle\bigcup_{i = 1}^{k} \mathsf{perms}(r_i)$. That is, all of the permissions assigned to $r_1, \ldots, r_k$ are assigned to $r$. If such roles exist, we create a new role $r^\prime$ and assign it permissions $\mathsf{perms}(r) \setminus \displaystyle\bigcup_{i = 1}^{k} \mathsf{perms}(r_i)$. Now, if $\mathsf{perms}(r^\prime) = \emptyset$, we simply dicard $r^\prime$, remove the role $r$ and assign every user who is assigned to $r$ to each of the roles $r_1, \ldots, r_k$. Otherwise, we remove the role $r$, add $r^\prime$ in its place, and assign every user who is assigned to $r$ to each of $r^\prime, r_1, \ldots, r_k$. We repeat till we reach a fixpoint, i.e., no such roles $r, r_1, \ldots, r_k$ exist. Thus, the number of roles is no greater than before, and we hope to have discovered a role $r$ that can be removed. Table \ref{tab:eneheuristic} shows the results for our hard instances. We show in that table also the error in the resultant size of the role-set against known upper-bounds. The error is calculated as: $(\text{\# roles determined by heuristic} - \text{known bound})/\text{known bound} \times 100$.

Note that these heuristic algorithms are not compatible with Algorithm \ref{alg:dom}: there is no way to meaningfully leverage information about dominators. Consequently, we run this heuristic directly on the original input access matrices.

An observation that stands out in Table \ref{tab:eneheuristic} is how fast the algorithms run: in at worst about 2 minutes total in our setup. Consequently, we do not hope to outdo these heuristics from the standpoint of axis (2) for trade-offs that we discuss at the start of this section, namely, time-efficiency in practice. Also, these algorithms do well along axis (3): they demonstrate tangible progress as they run. Thus, the only question with a new heuristic is whether we can improve along axis (1), the number of roles. That is, we ask whether we can improve on the ``\# roles, error \% (rounded)'' column in Table \ref{tab:eneheuristic}.

\begin{table*}
\caption{The results on the hard instances from the new benchmark for the heuristics of Ene et al.~\cite{ene08}. The results for the \# roles are for whichever greedy choice of user/permission of (i) smallest, or, (ii) largest degree yields the best results after the lattice-based postprocessing.}
\label{tab:eneheuristic}
\begin{tabular}{cccccc}
\toprule
\multirow{2}{*}{Instance} & \# roles, & \# roles, & known & \# roles, error \% & Time \\
& greedy & post-lattice & bound & (rounded) & (m:s)\\
\midrule
small 07 & 99 & 99 & 30 & 230 & $<$ 1s\\
medium 02 & 453 & 453 & 150 & 197 & $<$ 1s\\
medium 05 & 498 & 497 & 200 & 149 & $<$ 1s\\
comp 02.1 & 3650 & 3012 & 2000 & 51 & 0:14\\
comp 02.2 & 4709 & 3385 & 2000 & 69 & 0:19\\
comp 02.3 & 4945 & 3233 & 2000 & 62 & 0:20\\
comp 02.4 & 4369 & 3747 & 2000 & 87 & 0:23\\
comp 03.1 & 9017 & 5894 & 3000 & 96 & 0:45\\
comp 03.2 & 167,69 & 6821 & 3000 & 127 & 0:59\\
comp 03.3 & 151,49 & 9085 & 3000 & 203 & 1:11\\
comp 03.4 & 190,67 & 9105 & 3000 & 204 & 1:43\\
comp 04.1 & 8449 & 5445 & 3500 & 56 & 0:55\\
comp 04.2 & 9976 & 5763 & 3500 & 65 & 2:27\\
comp 04.3 & 7472 & 5478 & 3500 & 57 & 1:23\\
comp 04.4 & 9352 & 6000 & 3500 & 71 & 1:00\\
\bottomrule
\end{tabular}
\end{table*}

\paragraph*{Our heuristic} We leverage maximal biclique enumeration somewhat differently, as an earlier step to Algorithm \ref{alg:maxsets} from the previous section. Our approach here is to find large maximal bicliques, adopt roles that correspond to them, and remove the corresponding edges to shrink the input. We can do this while simultaneously checking whether our total number of maximal bicliques falls below the threshold to adopt Algorithm \ref{alg:maxsets}. The reason we deem this approach a heuristic is that it can lead to suboptimality --- a large maximal biclique is not necessarily part of an optimal solution set of maximal bicliques. We show our approach as Algorithm \ref{alg:hard}. Our implementation is more efficient that Algorithm \ref{alg:hard} suggests in that we can enumerate large maximal bicliques simultaneously as we count the total number of maximal bicliques as we generate them.

\LinesNumbered
\DontPrintSemicolon
\begin{algorithm}
\While{\# maximal bicliques $>$ threshold}{
    Enumerate maximal bicliques till we find a large one\;
    Adopt such a large maximal biclique as a role\;
    Mark all the edges of that maximal biclique as removed\;
}
Run Algorithm \ref{alg:maxsets}
\caption{Algorithm for hard instances}
\label{alg:hard}
\end{algorithm}

Algorithm \ref{alg:hard} mentions our threshold for the number of maximal bicliques for our approach from the previous section explicitly in the \textbf{while} condition of Line (1). It introduces another threshold somewhat implicitly --- what we mean by a large maximal biclique. This is subjective, and based on some trial-and-error, we set that threshold at 200 edges for our experiments; that is, at the time of removal, we remove a maximal biclique we have found only if it would remove at least 200 edges. The value of this threshold impacts the quality of our results, trade-off axis (1), and the time-efficiency, axis (2), in our discussions above on trade-offs.

For our empirical assessment, we have had to further subdivide the hard instances into ``hard'' and ``hardest'' instances. Correspondingly, our empirical results for the former are in Table \ref{tab:hard}, and for the latter are in Table \ref{tab:hardest}. The reason for this further subdivision is that we rely on enumerating maximal bicliques, and therefore, our algorithm needs to store and process adjacency information. For the hardest instances, which are comp 03.1 -- 03.4, our setup does not have sufficient memory to do this. Consequently, we further refine Algorithm \ref{alg:hard} as follows. We break the input up into what we call pieces, where each piece has the same number of users in the input. For example, comp 03.1 has 9992 users. We have broken this up into two instances, each of 4996 users, and their corresponding permissions. Our choice of users in each piece is arbitrary; a more informed choice may yield better results, an issue we leave for future work. Once we run Algorithm \ref{alg:hard} and have the sets of roles for each piece, we first simply combine them (i.e., add the respective numbers of roles) for an RBAC policy for the entire instance, and then run the lattice-based postprocessing of Ene et al.~\cite{ene08} that we describe above on the result to reduce the total number of roles.

\paragraph*{Takeaways} Our immediate observation from Tables \ref{tab:eneheuristic}, \ref{tab:hard} and \ref{tab:hardest} is that our heuristic is superior, sometimes vastly so, to the prior heuristic from the standpoint of axis (1) for the trade-offs, the number of roles. There is only one benchmark input for which it is worse: comp 03.4. Our other observation is that there is a trade-off along axis (2), time-efficiency. Whereas the prior heuristic runs within seconds, ours can take days. Whether this trade-off between the quality of output and time is worth it, is subject to deployment considerations, i.e., wherever such an algorithm is used. We think of role-mining as a kind of security-infrastructure process, and not a mission-critical process. Consequently, the trade-off may be worth incurring. Also, one can always run multiple algorithms over time, and pick the one that suits one's choice for these trade-offs. For the one benchmark input, comp 03.4 for which we perform worse, it has some structural characteristics, a deeper investigation of which we leave for future work.

\begin{table*}
\caption{The performance of Algorithm \ref{alg:hard} on the hard instances.}
\label{tab:hard}
\begin{tabular}{cccccc}
\toprule
Instance (after & \# roles, large & \# roles, & Error, \% vs.~known & Time, Alg.~\ref{alg:hard}, Lines & Time, total\\
Alg.~\ref{alg:dom}) & max.~bicliques & total & bound (rounded) & (1)--(4), h:m:s & h:m:s\\
\midrule
small 07 & 8 & 30 & 0 & 1:50:28 & 1:51:53\\
medium 02 & 61 & 150 & 0 & 1:27:47 & 1:29:48\\
medium 05 & 76 & 200 & 0 & 1:21:50 & 1:30:47\\
comp 02.1 & 47 & 2036 & 2 & 1:18:12 & 1:49:56\\
comp 02.2 & 56 & 2045 & 2 & 1:50:14 & 3:16:34\\
comp 02.3 & 417 & 2259 & 13 & 13:53:22 & 14:54:02\\
comp 02.4 & 360 & 2223 & 11 & 19:53:00 & 21:19:12\\
comp 04.1 & 158 & 3596 & 3 & 4:57:04 & 5:40:50\\
comp 04.2 & 215 & 3605 & 3 & 6:33:30 & 7:35:43\\
comp 04.3 & 219 & 3611 & 3 & 6:55:24 & 7:31:16\\
comp 04.4 & 353 & 3703 & 6 & 10:16:43 & 11:18:55\\
\bottomrule
\end{tabular}
\end{table*}

\begin{table*}
\caption{Performance on the hardest instances.}
\label{tab:hardest}
\begin{tabular}{cccccc}
\toprule
\multirow{2}{*}{Instance} & \multirow{2}{*}{\# pieces} & \# roles, & \# roles, & Error, \% vs.~known & \multirow{2}{*}{Time}\\
& & initial & final & bound (rounded) &\\
\midrule
comp 03.1 & 2 & 5646 & 3357 & 12 & $\approx$ 1 day\\
comp 03.2 & 3 & 9217 & 3313 & 10 & $\approx$ 4 days\\
comp 03.3 & 6 & 9885 & 5273 & 76 & $\approx$ 6 days\\
comp 03.4 & 6 & 10,293 & 10,290 & 243 & $\approx$ 6 days\\
\bottomrule
\end{tabular}
\end{table*}

\section{Conclusions}\label{sec:concl}

We have proposed new algorithms for the long-standing role-mining problem for a natural optimization objective: the minimization of the number of roles. We have observed that a prior efficient algorithm is useful as a first step, but does not scale to newer benchmark inputs. We have proposed algorithms based on what we call the enumeration of maximal bicliques. Our first algorithm, which enumerates such bicliques and then employs a corresponding reduction to Integer Linear Programming (ILP), addresses more than half of all the benchmark inputs. The notion of enumerating maximal bicliques provides us not only the heart of the algorithm, but also a meaningful measure of hardness for an input instance. For the hard instances, we propose a heuristic algorithm that leverages maximal biclique enumeration differently from the first. For the instances that are not hard, our approach gives exact results, i.e., a minimum-sized set of roles. For hard instances, our approach outperforms a prior heuristic for all but one benchmark input. Our work is backed up by an extensive empirical assessment, and we make all our code available as open-source.

\paragraph*{Future work} There is considerable scope for future work. One is to explore several avenues for improvements in performance. Our implementation is memory-bound, so it will be interesting to investigate whether we can enumerate maximal bicliques without having to store the entire information about adjacency at any given time. Also, our implementation is single-threaded; it will be interesting to explore concurrent or parallel implementations. Another topic for future work is another natural optimization objective for role mining: minimizing the number of edges in the output role-based policy. Prior work \cite{vaidya09} has shown that this objective is qualitatively different from the one we adopt in this work, of minimizing the number of roles. The work of Ene et al.~\cite{ene08} introduces this objective as well, but does not explicitly tackle it, except to point out that it is also \cnph, and a decision version is in the hardest subclass of \cnp from the standpoint of efficient approximability. It is known that biclique cover has numerous other applications, such as tiling in databases. It would be interesting to explore whether our approaches are effective for benchmark and/or real-world inputs for those problems.

\section*{Acknowledgements}

We thank Ian Molloy for making the original suite of benchmarks from the work of Ene et al.~\cite{ene08} available to us. We thank Puneet Gill for early discussions and for checking with Ene et al.~\cite{ene08} as to whether an implementation is available. We thank Yuantao Ji for the running example of Figure \ref{fig:irreducible:am}, that we use in this paper.

\bibliographystyle{plain}
\bibliography{mrm}

\end{document}